\documentclass[a4paper]{amsart}
\usepackage{amsthm}

\usepackage{amsmath}
\usepackage{amssymb}
\usepackage{tikz}
\usepackage{color}
\usetikzlibrary{matrix}
\usepackage{dsfont}
\usepackage{graphicx} 
\usepackage[margin=1.2 in]{geometry}
\usepackage{stackrel}
\numberwithin{equation}{section}
\newcommand\norm[1]{\left\lVert#1\right\rVert}

\title{Unitary induced channels and Tsirelson's problem} 
\author{Micha{\l} Banacki}
\address{International Centre for Theory of Quantum Technologies (ICTQT), University of Gda\'{n}sk, Jana Ba\.{z}y\'{n}skiego 1A, 80-309 Gda\'{n}sk, Poland}
\address{Institute of Mathematics, Faculty of Mathematics, Physics and Informatics, University of Gda\'{n}sk, Wita Stwosza 57, 80-308, Gda\'{n}sk, Poland}
\email{michal.banacki@ug.edu.pl}
\author{Paweł Horodecki}
\address{International Centre for Theory of Quantum Technologies (ICTQT), University of Gda\'{n}sk, Jana Ba\.{z}y\'{n}skiego 1A, 80-309 Gda\'{n}sk, Poland}
\address{Faculty of Applied Physics and Mathematics, Gda\'{n}sk University of Technology, Gabriela Narutowicza 11/12, 80-233 Gda\'{n}sk, Poland} 
\email{pawel.horodecki@ug.edu.pl}

\theoremstyle{plain}
\newtheorem{theorem}{Theorem}[section]

\newtheorem{corollary}[theorem]{Corollary}

\theoremstyle{definition}
\newtheorem{definition}[theorem]{Definition} 
 
\theoremstyle{remark}
\newtheorem{remark}[theorem]{Remark}

\newcommand{\beq}{\begin{equation}}
\newcommand{\eeq}{\end{equation}}
\newcommand{\be}{\begin{eqnarray}}
\newcommand{\ee}{\end{eqnarray}}
\newcommand{\beg}{\begin{eqnarray*}}
\newcommand{\eeg}{\end{eqnarray*}}

\begin{document}
\maketitle

\begin{abstract}Motivated by a recent progress concerning quantum commuting and quantum tensor models of composed systems we investigate a notion of (generalized) unitary induced quantum channel. Using properties of Brown algebras we provide an equivalent characterization of discussed families in both tensor and commuting paradigms. In particular, we provide an equivalent formulation of Tsirelson's conjecture (Connes' embedding problem) in terms of considered paradigms based on protocols which do not require measurements performed on infinite-dimensional subsystems. As a result we show that there is a difference between quantum commuting and quantum tensor models for generalized unitary induced channels.
\end{abstract}
\section{Preliminaries on quantum commuting and quantum tensor models}

In full generality quantum theory allows for two (possibly different) descriptions of composed systems consisting of two (separated) parts \cite{Boris, Scholz}. In the typical approach the whole system is represented by a tensor products of Hilbert spaces $H_A\otimes H_B$ while each subsystem corresponds to either $H_A$ or $H_B$ and local observables are realized as operators acting on $H_A$ and $H_B$ respectively. In the second and less prevalent approach inspired by algebraic quantum field theory \cite{Haag, Haag0}, the whole system corresponds to some Hilbert space $H$, while local observables for both subsystems are distinguished as families of commuting operators acting on $H$.

Specifically, for the bipartite Bell-type scenario (with $m\geq 2$ measurement settings and $n\geq 2$ possible outcomes per setting), depending on allowed paradigms we can consider different sets of possible behaviours $P=\left\{p(ab|xy)\right\}_{a,b,x,y}$ that provide probabilistic description of this measurement scheme. Quantum commuting correlation ($P\in \mathcal{C}_{qc}(m,k)$) are defined by the existence of a Hilbert space $H$, a unit vector $|\psi\rangle\in H$ and PVMs (POVMs) elements $P_{a|x},Q_{b|y}\in B(H)$ such that $ [P_{a|x},Q_{b|y}]=0$ and $p(ab|xy)=\langle\psi|P_{a|x}Q_{b|y}|\psi\rangle$, while quantum correlations ($P\in \mathcal{C}_{q}(m,k)$) are characterized by the existence of Hilbert spaces $H_A, H_B$ with $\dim H_A,\dim H_B<\infty$, a unit vector $|\psi\rangle\in H_A\otimes H_B$ and PVMs (POVMs) elements $P_{a|x}\in B(H_A)$, $Q_{b|y}\in B(H_B)$ such that $p(ab|xy)=\langle\psi|P_{a|x}\otimes Q_{b|y}|\psi\rangle$. With the set of quantum spatial correlations $\mathcal{C}_{qs}(m,k)$ described by a relaxation of previous assumption of finite-dimensionality, one obtains the following chain of (in general nontrivial - see \cite{Slofstra1, Slofstra2, Col, Col2, Prakash, Re}) inclusions
\begin{equation}\label{hier}
\mathcal{C}_q(m,n)\subset \mathcal{C}_{qs}(m,n) \subset\overline{\mathcal{C}_{q}(m,n)}=\mathcal{C}_{qa}(m,n)\subset 
\mathcal{C}_{qc}(m,n).
\end{equation}

It is convenient to compare sets of correlations obtainable in quantum commuting an quantum tensor frameworks using the notion of universal algebras of PVMs (compare \cite{Blackadar,Blackadar2} for the concept of universal C*-algebras defined by generators and relations).

\begin{definition}\label{uni}Universal algebra of PVMs $C^*(m,n)$ is a universal (unital) C*-algebra generated by a collection of elements $\left\{P_{a|x}\right\}_{a,x=1}^{n,m}$ 
forming $m$ projection-valued measures with $n$ outcomes, i.e.
\begin{equation}\nonumber
C^*(m,n)=C^*\Biggl \langle\left\{P_{a|x}\right\}_{a,x=1}^{n,m}: \sum_{a=1}^nP_{a|x}=\mathds{1}, P_{a|x}=P_{a|x}^*=P_{a|x}^2\biggr \rangle.
\end{equation}
\end{definition}
To describe a measurement scenario for a bipartite system consider two copies of universal algebra of PVMs $C_A^*(m,n)$ and $C_B^*(m,n)$ generated (as in Definition \ref{uni}) respectively by PVMs elements $P_{a|x}$ and $Q_{b|y}$. Observe that a behaviour $P=\left\{p(ab|xy)\right\}_{a,b,x,y}$ fulfills $P\in \mathcal{C}_{qc}(m,n)$ (respectively $P\in \mathcal{C}_{qa}(m,n)$) if and only if there exists a state $\phi\in S(C_A^*(m,n)\otimes_{max} C_B^*(m,n))$ (respectively $\phi\in S(C_A^*(m,n)\otimes_{min} C_B^*(m,n))$) such that $p(ab|xy)=\phi(P_{a|x}\otimes Q_{b|y})$.

Celebrated Tsirelson's problem \cite{Boris, Scholz} conjectures that $\mathcal{C}_{qa}(m,n)=\mathcal{C}_{qc}(m,n) $ for all scenarios with $m\geq 2$ and $n\geq 2$. This particular conjecture has been recently refuted \cite{Re} giving a negative answer to many long-standing open problem in theory of operator algebras (see \cite{Kirchberg, Fritz, Junge, Ozawa, Connes, Musat, Dykema, Harris1, Harris2, Kavruk} for various equivalent formulation of Tsirelson's conjecture).

Despite this remarkable result not all of its consequences (both physical and purely mathematical) have been thoroughly investigated. Specifically, construction of particular correlations (realizable in a physically justifiable model) within commuting paradigm yet beyond the closure of the set of quantum correlations remains an important problem. Similarly, beyond the strict setting of Bell inequalities, comparison of quantum tensor and quantum commuting models continues to be a vibrant area of research, for example in the context of steering scenarios \cite{Banacki, Nava, Yan}, quantum embezzlement tasks \cite{Cleve0, Cleve, Luijk1, Luijk2, Luijk3} and many other related topics \cite{Col, Col2, Gao, Goildbring1, Goldbring2, Luijk0,Paddock}.  

In this manuscript, following the general premise described above, we analyze quantum informational protocols based on (generalized) unitary induced channels that can be considered within both quantum commuting and quantum tensor paradigms without a need of measurements performed on local infinite-dimensional subsystems. We provide equivalent characterization of such sets of channels (in terms of so-called Brown algebras) and discuss their connection with the original Tsirelson's problem.

\section{Unitary induced channels}
 Through the paper, we assume that each subsystem of interest can be modeled by some (unital) C*-algebra $C$, while states $\phi$ of such subsystems are given by positive and unital linear functionals on $C$, and by $S(C)$ we denote the set of all such states on $C$. In particular for any Hibert space $H$ we will consider a unital C*-algebra $B(H)$ of all bounded operators acting on $H$. Note that in a light of previous discussion, separated subsystems $A$, $B$ can be characterized either by C*-subalgebras  $A\subset B(H_A)$, $B\subset B(H_B)$ (which can be seen as commuting subalgebras in $B(H_A\otimes H_B)$) or a pair of commuting C*-subalgebras $A,B\subset B(H)$.

Let two (separated) parties, say Alice and Bob, share a join system $AB$ and each of them has an access only to the separated part of that composed system (denoted respectively by $A$ and $B$). Moreover, let both observer have additional access to local, finite-dimensional ancillas $A'$ and $B'$ (of fixed dimension $n$, i.e. modeled by a C*-algebra of matrices $A',B'=M_n(\mathbb{C})$) such that subsystems $A'$ and $A$ (respectively $B'$ and $B$) may interact (through some unitary). 

In such a setup we consider an evolution of a state $\varphi_{A'B'}$ describing subsystem $A'B'$ under  interactation with subsystem $AB$ (described by a fixed state $\phi_{AB}$) mediated through given unitaries $U_{A'A}$ and $V_{BB'}$. Note that for a tensor product description we have $U_{A'A}\in A'\otimes B(H_A)$ and $V_{BB'}\in B(H_B)\otimes B'$, while $A'\otimes B(H_A\otimes H_B)\otimes B'$ stands for the C*-algebra of the whole system. On the other hand, in a commuting paradigm $U_{A'A}\in A'\otimes B(H)$ and $V_{BB'}\in B(H)\otimes B'$ such that they commute in the C*-algebra of the whole system given as $A'\otimes B(H)\otimes B'$.

Mathematical description of the above scenario depending on a choice of initial paradigm (i.e. quantum commuting or quantum tensor description of subsystem $AB$) is encapsulated by the following definitions of unitary induced channels.

\color{black}

\begin{definition}\label{def_qc}
We say that a map $\Lambda:S(M_n(\mathbb{C})\otimes M_n(\mathbb{C}))\rightarrow S(M_n(\mathbb{C})\otimes M_n(\mathbb{C}))$ defines a unitary induced channel in a quantum commuting model if there exists a state $\phi \in S(B(H))$ (for some Hilbert space $H$) and unitaries $U=(u_{ij})_{i,j}\in M_n(B(H))$, $V=(v_{kl})_{k,l}\in M_n(B(H))$ such that $(U\otimes \mathds{1}_{B'})(\mathds{1}_{A'}\otimes V)=(\mathds{1}_{A'}\otimes V)(U\otimes \mathds{1}_{B'})$ and 
\begin{equation}
\Lambda(\varphi_{A'B'})(g\otimes h)=\varphi_{A'B'}\otimes \phi((U\otimes \mathds{1}_{B'})(\mathds{1}_{A'}\otimes V)g\otimes \mathds{1}_{AB}\otimes h (\mathds{1}_{A'}\otimes V^*)(U^*\otimes \mathds{1}_{B'})) 
\end{equation}for any $\varphi_{A'B'}\in S(M_n(\mathbb{C})\otimes M_n(\mathbb{C}))$ and $g\otimes h \in A'\otimes B'=M_n(\mathbb{C})\otimes M_n(\mathbb{C})$.

By $\mathcal{L}_{qc}(n)$ we will denote the set of all unitary induced channels in a quantum commuting model for fixed $n$.
\end{definition}

\begin{definition}\label{def_qa}

We say that a map $\Lambda:S(M_n(\mathbb{C})\otimes M_n(\mathbb{C}))\rightarrow S(M_n(\mathbb{C})\otimes M_n(\mathbb{C}))$ defines a unitary induced channel in a quantum tensor model if there exists a state $\phi \in S(B(H_A\otimes H_B))$ (for some Hilbert spaces $H_A$, $H_B$) and unitaries  $U=(u_{ij})_{i,j}\in M_n(B(H_A))$, $V=(v_{kl})_{k,l}\in M_n(B(H_B))$ such that  
\begin{equation}
\Lambda(\varphi_{A'B'})(g\otimes h)=\varphi_{A'B'}\otimes \phi((U\otimes \mathds{1}_{BB'})(\mathds{1}_{A'A}\otimes V)g\otimes \mathds{1}_{AB}\otimes h (\mathds{1}_{A'A}\otimes V^*)(U^*\otimes \mathds{1}_{BB'})). 
\end{equation}for any $\varphi_{A'B'}\in S(M_n(\mathbb{C})\otimes M_n(\mathbb{C}))$ and $g\otimes h \in A'\otimes B'=M_n(\mathbb{C})\otimes M_n(\mathbb{C})$.

By $\mathcal{L}_{qa}(n)$ we will denote the set of all unitary induced channels in a quantum tensor model for fixed $n$.
\end{definition}

Allowing for additional dependency of unitaries $U,V$ on classical labels $x,y$ (i.e. allowing different types of possible interactions between subsystems $A'$ and $A$ or $B'$ and $B$), we may consider a generalized scenario for unitary induced channels.

\begin{definition}\label{def_qcxy}
We say that a collection $\left\{\Lambda_{xy}\right\}_{x,y}$ of maps $\Lambda_{xy}:S(M_n(\mathbb{C})\otimes M_n(\mathbb{C}))\rightarrow S(M_n(\mathbb{C})\otimes M_n(\mathbb{C}))$ defines a generalized unitary induced channel in a quantum commuting model if there exists a state $\phi \in S(B(H))$ (for some Hilbert space $H$) and unitaries  $U^x=(u^x_{ij})_{i,j}\in M_n(B(H))$, $V^y=(v^y_{kl})_{k,l}\in M_n(B(H))$ such that $(U^x\otimes \mathds{1}_{B'})(\mathds{1}_{A'}\otimes V^y)=(\mathds{1}_{A'}\otimes V^y)(U^x\otimes \mathds{1}_{B'})$  and
\begin{equation}
\Lambda_{xy}(\varphi_{A'B'})(g\otimes h)=\varphi_{A'B'}\otimes \phi ((U^x\otimes \mathds{1}_{B'})(\mathds{1}_{A'}\otimes V^y)g\otimes \mathds{1}_{AB}\otimes h (\mathds{1}_{A'}\otimes (V^y)^*)((U^x)^*\otimes \mathds{1}_{B'})).
\end{equation}for all $x,y=1,\ldots, m$ and for any $\varphi_{A'B'}\in S(M_n(\mathbb{C})\otimes M_n(\mathbb{C}))$, and $g\otimes h \in A'\otimes B'=M_n(\mathbb{C})\otimes M_n(\mathbb{C})$.

By $\mathcal{L}_{qc}(m,n)$ we will denote the set of all generalized unitary induced channels in a quantum commuting model for fixed $m,n$.
\end{definition}

\begin{definition}\label{def_qaxy}
We say that a collection $\left\{\Lambda_{xy}\right\}_{x,y}$ of maps $\Lambda_{xy}:S(M_n(\mathbb{C})\otimes M_n(\mathbb{C}))\rightarrow S(M_n(\mathbb{C})\otimes M_n(\mathbb{C}))$ defines a generalized unitary induced channel in a quantum tensor model if there exists a $\phi \in S(B(H_A\otimes H_B))$ (for some Hilbert spaces $H_A$, $H_B$) and unitaries  $U^x=(u^x_{ij})_{i,j}\in M_n(B(H_A))$, $V^y=(v^y_{kl})_{k,l}\in M_n(B(H_B))$ such that 
\begin{equation}
\Lambda_{xy}(\varphi_{A'B'})(g\otimes h)=\varphi_{A'B'}\otimes \phi ((U^x\otimes \mathds{1}_{BB'})(\mathds{1}_{A'A}\otimes V^y)g\otimes \mathds{1}_{AB}\otimes h (\mathds{1}_{A'A}\otimes (V^y)^*)((U^x)^*\otimes \mathds{1}_{BB'})).
\end{equation}for all $x,y=1,\ldots, m$ and for any $\varphi_{A'B'}\in S(M_n(\mathbb{C})\otimes M_n(\mathbb{C}))$, and $g\otimes h \in A'\otimes B'=M_n(\mathbb{C})\otimes M_n(\mathbb{C})$.

By $\mathcal{L}_{qc}(m,n)$ we will denote the set of all generalized unitary induced channels in a quantum tensor model for fixed $m,n$.
\end{definition}

Note that a state $\phi\in B(H)$ in Definition \ref{def_qc} and Definition \ref{def_qcxy} without loss of generality can be given as $\phi(\cdot)=\langle \eta, \cdot\ \eta \rangle$ for some unit vector $|\eta\rangle \in H$, while a state $\phi\in B(H_A\otimes H_B)$ in Definition \ref{def_qa} and Definition \ref{def_qaxy} can be expressed as $\phi(\cdot)=\langle \eta, \cdot\ \eta \rangle$ for some unit vector $|\eta\rangle \in H_A\otimes H_B$ without loss of generality up to the closure. Therefore, we can see sets of generalized unitary induced channels in analogy with sets $\mathcal{C}_{qc}(m,n)$ and $\mathcal{C}_{qa}(m,n)$.

\color{black}

Let us observe that any unitary $U\in M_n(B(H))=M_n(\mathbb{C})\otimes B(H)$ can be represented as 
\begin{equation}
 U= \sum_{i,j=1}^n E_{ij}\otimes u_{ij}
\end{equation}where $\left\{E_{ij}\right\}_{i,j=1}^{n}$ stand for matrix units and $\left\{u_{ij}\right\}_{i,j=1}^{n}$ are some elements in $B(H)$. Therefore, as has been shown in \cite{Cleve0}, commuting model can be completely characterized in terms of commutation relation between entries $u_{ij}$ and $v_{kl}$. Indeed, if $(U\otimes \mathds{1}_{B'})(\mathds{1}_{A'}\otimes V)=(\mathds{1}_{A'}\otimes V)(U\otimes \mathds{1}_{B'})$, then $(\mathds{1}_{A'}\otimes V)(U^*\otimes \mathds{1}_{B'})=(U^*\otimes \mathds{1}_{B'})(\mathds{1}_{A'}\otimes V)$ and therefore 
\begin{equation}\nonumber
\sum_{i,j,k,l=1}^n E_{ij}\otimes u_{ij}v_{kl}\otimes E_{kl}=(U\otimes \mathds{1}_{B'})(\mathds{1}_{A'}\otimes V)=(\mathds{1}_{A'}\otimes V)(U\otimes \mathds{1}_{B'})=\sum_{i,j,k,l=1}^n E_{ij}\otimes v_{kl}u_{ij}\otimes E_{kl}
\end{equation}
\begin{equation}\nonumber
\sum_{i,j,k,l=1}^n E_{ji}\otimes u^*_{ij}v_{kl}\otimes E_{kl}=(U^*\otimes \mathds{1}_{B'})(\mathds{1}_{A'}\otimes V)=(\mathds{1}_{A'}\otimes V)(U^*\otimes \mathds{1}_{B'})=\sum_{i,j,k,l=1}^n E_{ji}\otimes v_{kl}u^*_{ij}\otimes E_{kl}
\end{equation}or equivalently $u_{ij}v_{kl}=v_{kl}u_{ij}$ and $u^*_{ij}v_{kl}=v_{kl}u^*_{ij}$ for all $i,j,k,l=1,\ldots, m$. It is obvious, that the converse implication holds as well.

This simple observation allows for introduction of Brown algebras \cite{Brown_alg} as a convenient tool for characterisation of considered schemes.

\begin{definition}\label{Browndef}Brown algebra $\mathcal{U}(n)$ is a universal (unital) C*-algebra generated by a collection of elements $\left\{u_{ik}\right\}_{i,k=1}^{n}$ fulfilling relations making $U=(u_{ik})_{i,k}$ a unitary matrix, i.e.
\begin{equation}\nonumber
\mathcal{U}(n)=C^*\Biggl \langle\left\{u_{ik}\right\}_{i,k=1}^{n}: \sum_{k=1}^nu_{ik}u^*_{lk}=\delta_{il}\mathds{1}, \sum_{i=1}^nu^*_{ik}u_{il}=\delta_{kl}\mathds{1}\biggr \rangle.
\end{equation}
\end{definition}

Evoked notion of the Brown algebras enable us to provide a different description of (generalized) unitary induced channels.

\begin{theorem}\label{max}The following conditions are equivalent
\begin{enumerate}
\item $\Lambda\in \mathcal{L}_{qc}(n)$.
\item  There exists a state $\tilde{\phi} \in S(A\otimes_{max}B)$ for some unital C*-algebras $A,B$ such that 
for any state $\varphi_{A'B'}\in S(M_n(\mathbb{C})\otimes M_n(\mathbb{C}))$ and $g\otimes h\in A'\otimes B'= M_n(\mathbb{C})\otimes M_n(\mathbb{C})$

\begin{equation}\label{formula_1}
\Lambda(\varphi_{A'B'})(g\otimes h)=\sum_{i,j,k,l,p,r,s,t}\varphi_{A'B'}(E_{ij}gE_{kl}\otimes E_{pr}hE_{st})\tilde{\phi}(\tilde{u}_{ij}\tilde{u}^*_{lk}\otimes \tilde{v}_{pr}\tilde{v}^*_{ts}),
\end{equation}where $U=(\tilde{u}_{ij})_{i,j}\in M_n(A)$ and $V=(\tilde{v}_{kl})_{k,l}\in M_n(B)$ are unitaries.

\item  There exists a state  $\phi \in S(\mathcal{U}(n)\otimes_{max}\mathcal{U}(n))$ such that 
for any state $\varphi_{A'B'}\in S(M_n(\mathbb{C})\otimes M_n(\mathbb{C}))$ and $g\otimes h\in A'\otimes B'=M_n(\mathbb{C})\otimes M_n(\mathbb{C})$
\begin{equation}\label{formula_2}
\Lambda(\varphi_{A'B'})(g\otimes h)=\sum_{i,j,k,l,p,r,s,t}\varphi_{A'B'}(E_{ij}gE_{kl}\otimes E_{pr}hE_{st})\phi(u_{ij}u^*_{lk}\otimes v_{pr}v^*_{ts}),
\end{equation}where $\left\{u_{ij}\right\}_{i,j=1}^{n}$ and $\left\{v_{kl}\right\}_{k,l=1}^{n}$ stand respectively for generators of two copies of Brown algebra $\mathcal{U}(n)$.
\end{enumerate}
\end{theorem}
\begin{proof}To show implication $(1)\Rightarrow (2)$, observe that for $\Lambda\in \mathcal{L}_{qc}(n)$ we have
\begin{align}\label{id}
 \Lambda(\varphi_{A'B'})(g\otimes h)&=\varphi_{A'B'}\otimes \psi ((U\otimes \mathds{1}_{B'})(\mathds{1}_{A'}\otimes V)g\otimes \mathds{1}_{AB}\otimes h (\mathds{1}_{A'}\otimes V^*)(U^*\otimes \mathds{1}_{B'}))\\ \nonumber
&=\sum_{i,j,k,l,p,r,s,t}\varphi_{A'B'}(E_{ij}gE_{kl}\otimes E_{pr}hE_{st})\psi(\tilde{u}_{ij}\tilde{v}_{pr}\tilde{v}^*_{ts}\tilde{u}^*_{lk})\\ \nonumber
&=\sum_{i,j,k,l,p,r,s,t}\varphi_{A'B'}(E_{ij}gE_{kl}\otimes E_{pr}hE_{st})\psi(\tilde{u}_{ij}\tilde{u}^*_{lk}\tilde{v}_{pr}\tilde{v}^*_{ts})
\end{align}for some $\psi \in S(B(H))$ and commuting unitaries $U=(\tilde{u}_{ij})_{i,j}\in M_n(B(H))$, $V=(\tilde{v}_{lk})_{k,l}\in M_n(B)$. Consider unital C*-algebras $A=C^*(\left\{\tilde{u}_{ij}\right\})\subset B(H)$ and $B=C^*(\left\{\tilde{v}_{ij}\right\})\subset B(H)$, i.e. let $A$ and $B$ be the smallest C*-subalgebras in $B(H)$ containing respectively elements $\tilde{u}_{ij}$ and $\tilde{v}_{ij}$ for $i,j=1,\ldots, n$. Define the following unital $*$-isomorphisms $I_A:A\rightarrow A\subset B(H)$ and $I_B:B\rightarrow B\subset B(H)$ such that $I_A(c)=c$ and $I_B(d)=d$ for all $c\in A$, $d\in B$. Since $u_{ij}v_{kl}=v_{kl}u_{ij}$ for all $i,j,k,l=1,\ldots,n$, $A$ and $B$ are commuting C*-subalgebras in $B(H)$ and due to the property of maximal tensor products of C*-algebras (see \cite{Brown}), there exists a unital completely positive map (in fact it is a $*$-homomorphism) $I:A\otimes_{max}B\rightarrow B(H)$ such that $I(c\otimes d)=cd$ for a all $c\in A$, $d\in B$. Because of that, formula (\ref{id}) becomes (\ref{formula_1}) with $\tilde{\phi}=\psi\circ I\in S(A\otimes_{max}B)$ which concludes the first implication.

To justify $(2)\Rightarrow (3)$, observe that there exist unique unital $*$-homomorphisms $\Phi_A:\mathcal{U}(n)\rightarrow A$ and $\Phi_B:\mathcal{U}(n)\rightarrow B$ such that $\Phi_A(u_{ij})=\tilde{u}_{ij}$ and $\Phi_B(v_{ij})=\tilde{v}_{ij}$  for any $i,j=1,\ldots, n$ respectively. By the property of the maximal tensor product of C*-algebras (see \cite{Brown}) there exists a unital $*$-homomorphism $\Phi:\mathcal{U}(n)\otimes_{max}\mathcal{U}(n)\rightarrow A\otimes_{max}B$ such that $\Phi(c\otimes d)=\Phi_A(c)\otimes \Phi_B(d)$ for all $c,d\in \mathcal{U}(n)$. If so, then equation (\ref{formula_2}) coincides with equation (\ref{formula_1}) when $\phi=\tilde{\phi}\circ \Phi\in S(\mathcal{U}(n)\otimes_{max}\mathcal{U}(n))$ and $(3)$ is fulfilled.

Finally, to show implication $(3)\Rightarrow (1)$, note that there exists a Hilbert space $H$ such that $\mathcal{U}(n)\otimes_{max} \mathcal{U}(n)\subset B(H)$, i.e. $\mathcal{U}(n)\otimes_{max} \mathcal{U}(n)$ can be seen as a unital C*-subalgebra in some $B(H)$. If so, then one can extend $\phi\in S(\mathcal{U}(n)\otimes_{max} \mathcal{U}(n))$ to the state $\psi\in S(B(H))$. It is clear, that $\psi$ can replace $\phi$ in formula (\ref{formula_2}), hence $\Lambda\in \mathcal{L}_{qc}(n)$ with $\tilde{u}_{ij}=u_{ij}\otimes \mathds{1}_{B}\in \mathcal{U}(n)\otimes_{max} \mathcal{U}(n)\subset B(H)$ and $\tilde{v}_{ij}=\mathds{1}_{A}\otimes v_{ij}\in \mathcal{U}(n)\otimes_{max} \mathcal{U}(n)\subset B(H)$ for all $i,j=1,\ldots, n$.

\end{proof}

Analogous reformulations can be obtained in the case of the quantum tensor model, if one exchange properties of the maximal tensor product of C*-algebras with properties of the minimal one (compare with \cite{Brown}).

\begin{theorem}\label{min}The following conditions are equivalent
\begin{enumerate}
\item $\Lambda\in \mathcal{L}_{qa}(n)$.
\item  There exists a state $\tilde{\phi} \in S(A\otimes_{min}B)$ for some unital C*-algebras $A,B$ such that 
for any state $\varphi_{A'B'}\in S(M_n(\mathbb{C})\otimes M_n(\mathbb{C}))$ and $g\otimes h\in A'\otimes B'=M_n(\mathbb{C})\otimes M_n(\mathbb{C})$

\begin{equation}\label{formula_1ss}
\Lambda(\varphi_{A'B'})(g\otimes h)=\sum_{i,j,k,l,p,r,s,t}\varphi_{A'B'}(E_{ij}gE_{kl}\otimes E_{pr}hE_{st})\tilde{\phi}(\tilde{u}_{ij}\tilde{u}^*_{lk}\otimes \tilde{v}_{pr}\tilde{v}^*_{ts}),
\end{equation}where $U=(\tilde{u}_{ij})_{i,j}\in M_n(A)$ and $V=(\tilde{v}_{kl})_{k,l}\in M_n(B)$ are unitaries.
\item  There exists a state  $\phi \in S(\mathcal{U}(n)\otimes_{min}\mathcal{U}(n))$ such that 
for any state $\varphi_{A'B'}\in S(M_n(\mathbb{C})\otimes M_n(\mathbb{C}))$ and $g\otimes h\in A'\otimes B'=M_n(\mathbb{C})\otimes M_n(\mathbb{C})$
\begin{equation}\label{formula_2ss}
\Lambda(\varphi_{A'B'})(g\otimes h)=\sum_{i,j,k,l,p,r,s,t}\varphi_{A'B'}(E_{ij}gE_{kl}\otimes E_{pr}hE_{st})\phi(u_{ij}u^*_{lk}\otimes v_{pr}v^*_{ts}),
\end{equation}where $\left\{u_{ij}\right\}_{i,j=1}^{n}$ and $\left\{v_{kl}\right\}_{k,l=1}^{n}$ stand respectively for generators of two copies of Brown algebra $\mathcal{U}(n)$.
\end{enumerate}
\end{theorem}

To use universal properties of Brown algebras $\mathcal{U}(n)$ for a similar description of generalized scenarios, we recall the notion of the (full) free product $A*B$ of C*-algebras $A, B$ (see \cite{Avi,Boca,Pisier} for broader discussion of that concept) and the following technical result due to \cite{Boca}.

\begin{theorem}\label{boca}Let $A$, $B$ and $C$ be unital C*-algebras. If $\Phi_A:A\rightarrow C$ and $\Phi_B:B\rightarrow C$ are unital and completely positive maps, then there exists a unital completely positive map $\Phi:A*B\rightarrow C$ such that $\Phi|_A=\Phi_A$ and $\Phi|_B=\Phi_B$, i.e. there exists a common extension of both maps $\Phi_A, \Phi_B$.
\end{theorem}

Coming form this result we can provide the following characterization.

\begin{theorem}\label{max_2}The following conditions are equivalent
\begin{enumerate}
\item $\left\{\Lambda_{xy}\right\}_{x,y}\in \mathcal{L}_{qc}(m,n)$.
\item  There exists a state $\tilde{\phi} \in S(A\otimes_{max}B)$ for some unital C*-algebras $A,B$ such that 
for any state $\varphi_{A'B'}\in S(M_n(\mathbb{C})\otimes M_n(\mathbb{C}))$ and $g\otimes h\in A'\otimes B'=M_n(\mathbb{C})\otimes M_n(\mathbb{C})$
\begin{equation}\label{formula_12}
\Lambda_{xy}(\varphi_{A'B'})(g\otimes h)=\sum_{i,j,k,l,p,r,s,t}\varphi_{A'B'}(E_{ij}gE_{kl}\otimes E_{pr}hE_{st})\tilde{\phi}(\tilde{u}^x_{ij}(\tilde{u}^x_{lk})^*\otimes \tilde{v}^y_{pr}(\tilde{v}^y_{ts})^*),
\end{equation}where $U^x=(\tilde{u}^x_{ij})_{i,j}\in M_n(A)$ and $V^y=(\tilde{v}^y_{kl})_{k,l}\in M_n(B)$ are unitaries for any $x,y=1,\ldots, m$.
\item  There exists a state $\phi \in S((*_{x=1}^m\mathcal{U}^x(n))\otimes_{max}(*_{y=1}^m\mathcal{U}^y(n)))$ such that for any state $\varphi_{A'B'}\in S(M_n(\mathbb{C})\otimes M_n(\mathbb{C}))$ and $g\otimes h\in A'\otimes B'=M_n(\mathbb{C})\otimes M_n(\mathbb{C})$
\begin{equation}\label{formula_22}
\Lambda_{xy}(\varphi_{A'B'})(g\otimes h)=\sum_{i,j,k,l,p,r,s,t}\varphi_{A'B'}(E_{ij}gE_{kl}\otimes E_{pr}hE_{st})\phi(u^x_{ij}(u^x_{lk})^*\otimes v^y_{pr}(v^y_{ts})^*),
\end{equation}where for any $x,y=1,\ldots, m$ elements $\left\{u^x_{ij}\right\}_{i,j=1}^{n}$ and $\left\{v^y_{kl}\right\}_{k,l=1}^{n}$ stand respectively for generators of the $x$-th copy of Brown algebra $\mathcal{U}^x(n)\subset *_{x=1}^m\mathcal{U}^x(n)$ and generators of the $y$-th copy of Brown algebra $\mathcal{U}^y(n)\subset *_{y=1}^m\mathcal{U}^y(n)$.
\end{enumerate}
\end{theorem}
\begin{proof} Implications $(1)\Rightarrow (2)$, as well as $(3)\Rightarrow (1)$  can be justified in similar manner to corresponding implications from Theorem \ref{max}.

Therefore, it remains to show that $(2)\Rightarrow (3)$. Observe that for any $x,y=1,\ldots, m$ there exist unique unital *-homomorphisms $\Phi^x_A:\mathcal{U}^x(n)\rightarrow A$ and $\Phi^y_B:\mathcal{U}^y(n)\rightarrow B$ such that $\Phi^x_A(u^x_{ij})=\tilde{u}^x_{ij}$ and $\Phi^y_B(v^y_{ij})=\tilde{v}^y_{ij}$  for any $i,j=1,\ldots, n$ respectively. If so, then recursive usage of Theorem \ref{boca} implies the existence of a unital completely positive maps $\Phi_A:*_{x=1}^m\mathcal{U}^x(n)\rightarrow A$ and $\Phi_B:*_{y=1}^m\mathcal{U}^y(n)\rightarrow B$ such that respectively for any $c \in \mathcal{U}^x(n)$ and $d \in \mathcal{U}^y(n)$ we have $\Phi_A(c)=\Phi_A^x(c)$ and $\Phi_B(d)=\Phi_B^y(d)$. Therefore, by the property of the maximal tensor product of C*-algebras, there exists a unital completely positive map $\Phi:(*_{x=1}^m\mathcal{U}^x(n))\otimes_{max}(*_{y=1}^m\mathcal{U}^y(n))\rightarrow A\otimes_{max}B$ such that $\Phi(c\otimes d)=\Phi_A(c)\otimes \Phi_B(d)$ for all $c \in *_{x=1}^m\mathcal{U}(n)$ and $d \in *_{y=1}^m\mathcal{U}(n)$. If so, then equation (\ref{formula_22}) coincides with equation (\ref{formula_12}) when $\phi=\tilde{\phi}\circ \Phi\in S((*_{x=1}^m\mathcal{U}^x(n))\otimes_{max}(*_{y=1}^m\mathcal{U}^y(n)))$ and the proof is completed.
\end{proof}

Once more, similar reasoning leads to the following result concerning the set of generalized unitary induced channels in quantum tensor paradigm.

\begin{theorem}\label{min_2}The following conditions are equivalent
\begin{enumerate}
\item $\left\{\Lambda_{xy}\right\}_{x,y}\in \mathcal{L}_{qa}(m,n)$.
\item  There exists a state $\tilde{\phi} \in S(A\otimes_{min}B)$ for some unital C*-algebras $A,B$ such that 
for any state $\varphi_{A'B'}\in S(M_n(\mathbb{C})\otimes M_n(\mathbb{C}))$ and $g\otimes h\in A'\otimes B'= M_n(\mathbb{C})\otimes M_n(\mathbb{C})$
\begin{equation}\label{formula_12uu}
\Lambda_{xy}(\varphi_{A'B'})(g\otimes h)=\sum_{i,j,k,l,p,r,s,t}\varphi_{A'B'}(E_{ij}gE_{kl}\otimes E_{pr}hE_{st})\tilde{\phi}(\tilde{u}^x_{ij}(\tilde{u}^x_{lk})^*\otimes \tilde{v}^y_{pr}(\tilde{v}^y_{ts})^*),
\end{equation}where $U^x=(\tilde{u}^x_{ij})_{i,j}\in M_n(A)$ and $V^y=(\tilde{v}^y_{kl})_{k,l}\in M_n(B)$ are unitaries for any $x,y=1,\ldots, m$.
\item  There exists a state $\phi \in S((*_{x=1}^m\mathcal{U}^x(n))\otimes_{min}(*_{y=1}^m\mathcal{U}^y(n)))$ such that for any state $\varphi_{A'B'}\in S(M_n(\mathbb{C})\otimes M_n(\mathbb{C}))$ and $g\otimes h\in A'\otimes B'= M_n(\mathbb{C})\otimes M_n(\mathbb{C})$
\begin{equation}\label{formula_22uu}
\Lambda_{xy}(\varphi_{A'B'})(g\otimes h)=\sum_{i,j,k,l,p,r,s,t}\varphi_{A'B'}(E_{ij}gE_{kl}\otimes E_{pr}hE_{st})\phi(u^x_{ij}(u^x_{lk})^*\otimes v^y_{pr}(v^y_{ts})^*),
\end{equation}where for any $x,y=1,\ldots, m$ elements $\left\{u^x_{ij}\right\}_{i,j=1}^{n}$ and $\left\{v^y_{kl}\right\}_{k,l=1}^{n}$ stand respectively for generators of the $x$-th copy of Brown algebra $\mathcal{U}^x(n)\subset *_{x=1}^m\mathcal{U}^x(n)$ and generators of the $y$-th copy of Brown algebra $\mathcal{U}^y(n)\subset *_{y=1}^m\mathcal{U}^y(n)$.
\end{enumerate}
\end{theorem}

\vspace{0,5cm}

Observe that since states $\varphi \in S(M_n(\mathbb{C})\otimes M_n(\mathbb{C}))$ are in bijective correspondence with density matrices $\rho\in (M_n(\mathbb{C})\otimes M_n(\mathbb{C}))_+$ one can (abusing notation) treat maps $\Lambda_{xy}$ considered above as quantum channels, i.e. completely positive and trace preserving maps $\Lambda_{xy}:M_n(\mathbb{C})\otimes M_n(\mathbb{C})\rightarrow M_n(\mathbb{C})\otimes M_n(\mathbb{C})$. Indeed, consider a density matrix $\rho=\sum_{a,b,c,d}\rho_{abcd}E_{ab}\otimes E_{cd}$ such that $\varphi(\cdot)=\mathrm{Tr}(\rho\ \cdot)$, then  $\Lambda_{xy}(\varphi)(\cdot)=\mathrm{Tr}(\Lambda_{xy}(\rho)\ \cdot)$ where (according to Theorem \ref{max_2} and Theorem \ref{min_2})
\begin{equation}\label{channel_form}
\Lambda_{xy}(\rho)=\sum_{j,k,r,s}\left[\sum_{i,l,p,t} \phi(u_{ij}^x(u_{lk}^x)^*\otimes v_{pr}^y(v_{ts}^y)^*)\rho_{litp}\right] E_{kj}\otimes E_{sr}
\end{equation}defines a completely positive and trace preserving map. Therefore, we will interchangeably use both meaning of $\Lambda_{xy}$ (or $\mathcal{L}_{qc/qa}(m,n)$) without explicitly stating it.

Before we discuss relationship of introduced paradigms with Tsirelson's problem, let us also consider a relaxed (state-dependent) scenario, when instead of comparison between generalized channels $\left\{\Lambda_{xy}\right\}_{x,y}$ one consider comparison between families of states (density matrices) $\left\{\Lambda_{xy}(\rho)\right\}_{x,y}$ obtainable from a fixed initial state (density matrix) $\rho$. This paradigm boils down to analysis of the following classes
\begin{equation}
I_{qc}(m,\rho)=\left\{\left\{\Lambda_{xy}(\rho)\right\}_{x,y}:\left\{\Lambda_{xy}\right\}_{x,y}\in \mathcal{L}_{qc}(m,n)\right\},
\end{equation}
\begin{equation}
I_{qa}(m,\rho)=\left\{\left\{\Lambda_{xy}(\rho)\right\}_{x,y}:\left\{\Lambda_{xy}\right\}_{x,y}\in \mathcal{L}_{qa}(m,n)\right\},
\end{equation}where we introduce a simplified notation $I_{qc}(\rho)=I_{qc}(1,\rho)$, $I_{qa}(\rho)=I_{qa}(1,\rho)$. Note that in particular $I_{qc}(m,\rho)\neq I_{qa}(m,\rho)$ may be seen as a witness of $\mathcal{L}_{qc}(m,n)\neq \mathcal{L}_{qc}(m,n)$.

Consider now an ancillary (input) state $\rho_{A'B'}$, a fixed resource state $\tilde{\rho}_{AB}$ such that 
$\rho_{A'B'},\tilde{\rho}_{AB}\in (M_n(\mathbb{C})\otimes M_n(\mathbb{C}))_+$ and unitaries $U_{A'A},V_{BB'}=U_{SWAP}$, where $U_{SWAP}$ stands for the swap operator. Note that a unitary induced channel $\Lambda$ (in a quantum tensor model) given by that data acts on $\rho_{A'B'}$ in the following way 
\begin{equation}
\Lambda(\rho_{A'B'})=\mathrm{Tr}_{AB}((U_{A'A}^*\otimes V_{BB'}^*)\rho_{A'B'}\otimes\tilde{\rho}_{AB}(U_{A'A}\otimes V_{BB'}))=\tilde{\rho}_{A'B'},
\end{equation}so clearly $I_{qc}(\rho)=I_{qa}(\rho)$ for any choice of $\rho$, as tensor model enable us to construct unitary induced channels transforming any fixed state to any fixed state. We will return to this simple example in Remark \ref{rem} and its following discussion.

\section{Connections with Tsirelson's conjecture}
 In the remaining part of the manuscript (based on characterisation in terms of Brown algebras $\mathcal{U}(n)$) we will show a connection between original Tsirelson's conjecture and generalized unitary induced channels.

In order to do this, we recall well-known notions of local lifting property and weak expectation property. Let $A$ be a C*-algebra. We will say that $A$ has LLP (i.e. $A$ has a local lifting property) if and only if $(A,B(\ell_2))$ is a nuclear pair, i.e. $A\otimes_{max}B(\ell_2)=A\otimes_{min}B(\ell_2)$ \cite{Pisier} (here $\ell_2$ stands for a separable Hilbert space). Similarly, we will say that $A$ has WEP (i.e. $A$ has weak expectation property) if and only if  $(A,C^*(\mathbb{F}_\infty))$ is a nuclear pair, where $C^*(\mathbb{F}_\infty)$ stands for full group C*-algebra of the free group $\mathbb{F}_\infty$ with countably many generators \cite{Pisier}.

The next theorem summarizes a collection of partial results presented in \cite{Fritz, Junge, Harris1, Kirchberg, Ozawa}  (see also \cite{Pisier}) and provides a bridge between Tsirelson's problem and (some of) equivalent statements stated in the language of operator algebras.

\begin{theorem}\label{Tsirelson}The following statements are equivalent
\begin{enumerate}
\item Connes' embedding conjecture is true.
\item Tsirelson's conjecture is true (i.e. $\mathcal{C}_{qa}(m,n)=
\mathcal{C}_{qc}(m,n)$ for all numbers of measurement settings $m\geq 2$ and outcomes $n\geq 2$).
\item  $\mathcal{C}_{qa}(m,2)=
\mathcal{C}_{qc}(m,2)$ for all $m\geq 2$.
\item $\mathcal{U}(n)\otimes_{max}\mathcal{U}(n)=\mathcal{U}(n)\otimes_{min}\mathcal{U}(n)$ for all $n\geq 2$.
\item $\mathcal{U}(n)\otimes_{max}\mathcal{U}(n)=\mathcal{U}(n)\otimes_{min}\mathcal{U}(n)$ for some $n\geq 2$.
\item Any C*-algebra $A$ with LLP has WEP.
\end{enumerate}
\end{theorem}

\begin{remark}\label{rem}
Interestingly, refutation of Tsirelson's conjecture \cite{Re} and the above theorem show that $I_{qa}(\rho)=I_{qc}(\rho)$ (for all density matrices $\rho\in M_n(\mathbb{C})\otimes M_n(\mathbb{C})_+$) despite the fact that (due to Theorem \ref{max} and Theorem \ref{min}) characterisation of both sets differs only by state spaces $S(\mathcal{U}(n)\otimes_{max}\mathcal{U}(n))$, $S(\mathcal{U}(n)\otimes_{min}\mathcal{U}(n))$, while $\mathcal{U}(n)\otimes_{max}\mathcal{U}(n)\neq \mathcal{U}(n)\otimes_{min}\mathcal{U}(n)$.
\end{remark}

Following this remark, it is in general non-trivial to consider the following operator systems (see \cite{Kavruk2} for an introduction to the theory of operator systems)

\begin{equation}
\mathcal{F}_{qc}^\rho=\mathrm{span}\left(\left\{\mathds{1}\otimes \mathds{1}\right\}\cup \left\{\sum_{l,i,t,p}\rho_{litp} u_{ij}u_{lk}^*\otimes v_{pr}v_{ts}^*:j,k,r,s=1,\ldots, n\right\} \right)\subset \mathcal{U}(n)\otimes_{max}\mathcal{U}(n),
\end{equation}

\begin{equation}
\mathcal{F}_{qa}^\rho=\mathrm{span}\left(\left\{\mathds{1}\otimes \mathds{1}\right\}\cup \left\{\sum_{l,i,t,p}\rho_{litp} u_{ij}u_{lk}^*\otimes v_{pr}v_{ts}^*:j,k,r,s=1,\ldots, n\right\} \right)\subset \mathcal{U}(n)\otimes_{min}\mathcal{U}(n),
\end{equation}defined for any density matrix $\rho=\sum_{a,b,c,d}\rho_{abcd}E_{ab}\otimes E_{cd}\in (M_n(\mathbb{C})\otimes M_n(\mathbb{C}))_+$ (as $\mathcal{F}_{qc}^\rho$ and $\mathcal{F}_{qa}^\rho$ may not be complete order isomorphic). Formula (\ref{channel_form}) together with equality $I_{qa}(\rho)=I_{qc}(\rho)$ enable us to immediately compare state spaces of both operator systems, throught the following corollary.

\begin{corollary}For any $n\geq 2$ and any density matrix $\rho\in (M_n(\mathbb{C})\otimes M_n(\mathbb{C}))_+$ we have $S(\mathcal{F}_{qc}^{\rho})=S(\mathcal{F}_{qa}^{\rho})$.
\end{corollary}

After this interlude related to the previous remark, we are ready to introduce the main theorem.

\begin{theorem}\label{thm_main_2}The following statements are equivalent  
\begin{enumerate}
\item Tsirelson's conjecture is true.
\item $\mathcal{L}_{qc}(m,n)=\mathcal{L}_{qa}(m,n)$ for all $n,m$.
\item For any $m,n$ and any $\left\{\Lambda_{xy}\right\}_{x,y}\in \mathcal{L}_{qc}(m,n)$ there exists $\left\{\Lambda'_{xy}\right\}_{x,y}\in \mathcal{L}_{qa}(m,n)$ such that $(\Lambda_{xy}(E_{kj}\otimes E_{sr}))_{kjsr}=(\Lambda'_{xy}(E_{kj}\otimes E_{sr}))_{kjsr}$ for all $k,j,s,r=1,\ldots, n$ and $x,y=1,\ldots, m$.
\item $\mathcal{L}_{qc}(m,2)=\mathcal{L}_{qa}(m,2)$ for all $m$.
\item For any $m$ and any $\left\{\Lambda_{xy}\right\}_{x,y}\in \mathcal{L}_{qc}(m,2)$ there exists $\left\{\Lambda'_{xy}\right\}_{x,y}\in \mathcal{L}_{qa}(m,2)$ such that $(\Lambda_{xy}(E_{kj}\otimes E_{sr}))_{kjsr}=(\Lambda'_{xy}(E_{kj}\otimes E_{sr}))_{kjsr}$ for all $k,j,s,r=1,2$ and $x,y=1,\ldots, m$.
\end{enumerate}
\end{theorem}
\begin{proof}
Observe that implications $(2)\Rightarrow (3)$, $(2) \Rightarrow (4)$, $(3)\Rightarrow  (5)$ and $(4)\Rightarrow  (5)$ are obvious.

We will show that $(1) \Rightarrow (2)$. Let us assume that Tsirelson's problem admits a positive solution. For any $m,n$ consider $\left\{\Lambda_{xy}\right\}_{x,y}\in \mathcal{L}_{qc}(m,n)$, then $\left\{\Lambda_{xy}\right\}_{x,y}$ can be repented according to point $(3)$ of Theorem \ref{max_2}. Since Brown algebra $\mathcal{U}(n)$ admits LLP \cite{Harris1}, due to Proposition 3.21 in \cite{Ozawa0} (based on \cite{Pisier96}), free product $*_{x=1}^m \mathcal{U}(n)$ of $m$ copies of $\mathcal{U}(n)$ has LLP as well. By the initial assumption and Theorem \ref{Tsirelson}, $*_{x=1}^m \mathcal{U}(n)$ also has WEP. As $A\otimes_{max}B=A\otimes_{min}B$ for all C*-algebras $A,B$ such  that $A$ has LLP and $B$ has WEP \cite{Pisier}, we infer that $\left\{\Lambda_{xy}\right\}_{x,y}\in \mathcal{L}_{qa}(m,n)$ due to Theorem \ref{min_2}. As on the other hand, $\mathcal{L}_{qa}(m,n)\subset \mathcal{L}_{qc}(m,n)$, the reasoning is completed.

To show $(5) \Rightarrow  (2)$, one can assume that $n=2$, although the reasoning remains the same for arbitrary $n$, and for that reason we will not restrict our notation to this specific case.

Consider generators $P_{a|x}\in C_A^*(m,n)$ and the following unitaries 
\begin{equation}
u_{a'}^x=\sum_{a=1}^ne^{\frac{2\pi i aa'}{n}}P_{a|x}
\end{equation}for $a'=1,\ldots, n$ and $x=1,\ldots, m$. Using inverse Fourier transform we reconstruct initial projectors as linear combinations of the above unitaries
\begin{equation}\label{fourier}
P_{a|x}=\frac{1}{n}\sum_{a'=1}^ne^{\frac{-2\pi i aa'}{n}}u^x_{a'}=\sum_{a'=1}^n c_{aa'}u^x_{a'}.
\end{equation}Analogously, for generators $Q_{b|y}\in C_B^*(m,n)$ we introduce unitaries $v^y_{b'}$ such that $Q_{b|y}=\sum_{b'}c_{bb'}v^y_{b'}$. 

We will use above unitaries to construct certain maps. Namely, for any $x,y=1,\ldots,m$ consider unital $*$-homomorphisms
$\pi_{x}:\mathcal{U}^x(n)\rightarrow C^*_A(m,n)$, $\tilde{\pi}_{y}:\mathcal{U}^y(n)\rightarrow C^*_B(m,n)$ given on generators $u_{ik}^x$, $v_{jl}^y$ as
\begin{equation}
\pi_{x}(u_{ik}^x)=u_i^x\delta_{ik},
\end{equation}
\begin{equation}
\tilde{\pi}_{y}(v_{jl}^y)=v_j^y\delta_{jl}.
\end{equation}Note that the existence of such maps follows directly from the universal property of the Brown algebra $\mathcal{U}(n)$. According to Theorem \ref{boca} there exist unital completely positive maps $\Psi_A:*_{x=1}^m\mathcal{U}^x(n)\rightarrow C^*_A(m,n)$ and $\Psi_B:*_{y=1}^m\mathcal{U}^y(n)\rightarrow C^*_B(m,n)$ such that $\Psi_A|_{\mathcal{U}^x(n)}=\pi_x$ and $\Psi_B|_{\mathcal{U}^y(n)}=\tilde{\pi}_y$ for all $x,y=1,\ldots,m$. Using relations (\ref{fourier}) for $P_{a|x}\in C^*_A(m,n)$ (and $Q_{b|y}\in C^*_B(m,n)$ respectively) we obtain
\begin{equation}
P_{a|x}=\sum_{a'=1}^nc_{aa'}u^x_{a'}=\sum_{a'=1}^nc_{aa'}\pi_x(u^x_{a'a'})=\pi_x\left(\sum_{a'=1}^nc_{aa'}u^x_{a'a'}\right)=\pi_x\left((\sum_{a'=1}^nc_{aa'}u^x_{a'a'})(\sum_{a'=1}^nc_{aa'}u^x_{a'a'})^*\right)
\end{equation}and analogously
\begin{equation}
Q_{b|y}=\tilde{\pi}_y\left((\sum_{b'=1}^nc_{bb'}v^y_{b'b'})(\sum_{b'=1}^nc_{bb'}v^y_{b'b'})^*\right).
\end{equation}Define positive operators $M_{a|x}=(\sum_{a'}c_{aa'}u^x_{a'a'})(\sum_{a'}c_{aa'}u^x_{a'a'})^*$ and $N_{b|y}=(\sum_{b'}c_{bb'}v^y_{b'b'})(\sum_{b'}c_{bb'}v^y_{b'b'})^*$.
Observe that due to particular form of coefficients $c_{aa'}$ one can show 
\begin{equation}
 \sum_{a=1}^nM_{a|x}=\frac{1}{n}\sum_{a=1}^nu^x_{aa}(u^x_{aa})^*
\end{equation}and since $\norm{u_{aa}}\leq 1$, so $\norm{ \sum_{a}M_{a|x}}\leq 1$ and $\sum_{a}M_{a|x}\leq \mathds{1}$. Similarly, we obtain $\sum_bN_{b|y}\leq \mathds{1}$.

Consider an arbitrary behaviour $P=\left\{p(ab|xy)\right\}_{a,b,x,y}$ such that $P\in \mathcal{C}_{qc}(m,n)$. Then for some $\varphi\in S(C^*_A(m,n)\otimes_{max}C_B^*(m,n))$ we have
\begin{equation}
p(ab|xy)=\varphi(P_{a|x}\otimes Q_{b|y})=\varphi(\Psi_A(M_{a|x})\otimes \Psi_B(N_{b|y}))=\tilde{\varphi}(M_{a|x}\otimes N_{b|y}),
\end{equation}where $\tilde{\varphi}=\varphi\circ (\Psi_A\otimes \Psi_B)\in S(*_{x=1}^m\mathcal{U}^x(n)\otimes_{max} *_{y=1}^m\mathcal{U}^y(n))$. 

Assume now (point $(5)$) that for any $\left\{\Lambda_{xy}\right\}_{x,y}\in \mathcal{L}_{qc}(m,n)$ there exists $\left\{\Lambda'_{xy}\right\}_{x,y}\in \mathcal{L}_{qa}(m,n)$ such that $(\Lambda_{xy}(E_{kj}\otimes E_{sr}))_{kjsr}=(\Lambda'_{xy}(E_{kj}\otimes E_{sr}))_{kjsr}$. Taking into account formula (\ref{channel_form}) direct calculations show that for any state $\phi \in S(*_{x=1}^m\mathcal{U}^x(n)\otimes_{max} *_{y=1}^m\mathcal{U}^y(n))$ there exists a state $\phi' \in S(*_{x=1}^m\mathcal{U}^x(n)\otimes_{min} *_{y=1}^m\mathcal{U}^y(n))$ such that
\begin{equation}
    \phi(u_{jj}^x(u_{kk}^x)^*\otimes v_{rr}^y(v_{ss}^y)^*)=\phi'(u_{jj}^x(u_{kk}^x)^*\otimes v_{rr}^y(v_{ss}^y)^*)
\end{equation}for all $j,k,s,r=1,\ldots, n$ and $x,y=1,\ldots, m$.

Since $M_{a|x}\in \mathrm{span}(\left\{u^x_{a'a'}(u^x_{b'b'})^*: a',b'=1,\ldots, n\right\})$, $N_{b|y}\in \mathrm{span}(\left\{v^y_{a'a'}(v^y_{b'b'})^*: a',b'=1,\ldots, n\right\})$ for all $a,b=1,\ldots, n$ and $x,y=1,\ldots, m$, there exists a state $\hat{\varphi}\in S(*_{x=1}^m\mathcal{U}^x(n)\otimes_{min} *_{y=1}^m\mathcal{U}^y(n))$ such that
\begin{equation}
p(ab|xy)=\hat{\varphi}(M_{a|x}\otimes N_{b|y}).
\end{equation}
Define POVMs elements $\hat{M}_{a|x}\in *_{x=1}^m\mathcal{U}^x(n)$ and $\hat{N}_{b|y}\in *_{y=1}^m\mathcal{U}^y(n)$ by $\hat{M}_{n|x}=M_{n|x}+(\mathds{1}-\sum_{a=1}^{n-1}M_{a|x})$, $\hat{N}_{n|y}=N_{n|y}+(\mathds{1}-\sum_{b=1}^{n-1}N_{b|y})$
and $\hat{M}_{a|x}=M_{a|x}$, $\hat{N}_{b|y}=N_{b|y}$ for $a,b=1,\ldots, n-1$, while putting $\hat{p}(ab|xy)=\hat{\varphi}(\hat{M}_{a|x}\otimes \hat{N}_{b|y})$. Using normalization $\sum_{a,b}p(ab|xy)=\sum_{a,b}\hat{p}(ab|xy)=1$ one can show that $p(ab|xy)=\hat{p}(ab|xy)=\hat{\varphi}(\hat{M}_{a|x}\otimes \hat{N}_{b|y})$, hence $P\in \mathcal{C}_{qa}(m,n)$. Since the choice of $m,n$ was arbitrary, we conclude that $\mathcal{C}_{qa}(m,n)=\mathcal{C}_{qc}(m,n)$ for all $m,n$ and the proof is completed (according to Theorem \ref{Tsirelson} it is enough to consider $n=2$).
\end{proof}

Note that the above result bears a twofold meaning. Observe that, due the groundbreaking observation \cite{Re} refuting Tsirelson's conjecture, our equivalence enable us to witness a Tsirelson's gap in a protocol which is independent of measurements performed on infinite-dimensional subsystems and relies only on comparison of input and output subsystems (with implicit use of finite-dimensional measurements for state tomography performed on finite-dimensional ancillas). It is (in principal) not necessary to perform measurements on infinite-dimensional systems in order to detect insufficiency of a tensor product framework, i.e.  $\mathcal{L}_{qc}(m,n)\neq \mathcal{L}_{qa}(m,n)$ for some $m,n$.

On the other hand, in the light of Theorem \ref{thm_main_2}, an explicit construction of generalized unitary induced channels within quantum commuting paradigm (for some choice of $m,n$), yet without quantum tensor model, would provide a simpler proof of a negative answer to Tsirelson's question in a form of a concrete counterexample with clearer physical meaning.

We will conclude with yet another equivalence based on the same general idea as a proof of Theorem \ref{thm_main_2}.

\begin{corollary} Tsirelson's conjecture is true if and only if
for any $m,n$ (respectively for any $m$ and $n=2$) and any $\left\{\Lambda_{xy}\right\}_{x,y}\in \mathcal{L}_{qc}(m,n)$ there exists $\left\{\Lambda'_{xy}\right\}_{x,y}\in \mathcal{L}_{qa}(m,n)$ such that 
\begin{equation}\label{last_cond}
\sum_{k,j,s,r}c_{aj}\overline{c_{ak}}c_{br}\overline{c_{bs}}(\Lambda_{xy}(E_{kj}\otimes E_{sr}))_{kjsr}=\sum_{k,j,s,r}c_{aj}\overline{c_{ak}}c_{br}\overline{c_{bs}}(\Lambda'_{xy}(E_{kj}\otimes E_{sr}))_{kjsr}
\end{equation}
 for all $a,b=1,\ldots ,n$ (respectively $a,b=1,2$)  and $x,y=1,\ldots, m$, while $c_{aj}$ are coefficients defined in the proof of Theorem \ref{thm_main_2}.
\end{corollary}

\color{black}

\section*{Acknowledgements}
This work was supported by the National Science Centre, Poland through grant Maestro (2021/42/A/ST2/00356). The authors are grateful to Marcin Marciniak, Ryszard Paweł Kostecki and Michał Eckstein for fruitful discussions.

\end{document}